\documentclass{llncs}
\usepackage{amsmath, amssymb, stmaryrd, bbm, bussproofs,xspace}
\usepackage{multicol}
\usepackage[matrix,arrow,curve]{xy}
\usepackage{textcomp}
\usepackage{times}
\usepackage[author=anonymous,margin]{fixme}
\FXRegisterAuthor{dg}{adg}{DG}
\FXRegisterAuthor{ls}{als}{LS}

\newcommand{\Sem}[1]{{[\![#1]\!]}}
\newcommand{\impl}{\Longrightarrow\xspace}

\newcommand{\hearts}{\heartsuit}

\renewcommand{\theta}{\vartheta}

\newcommand{\lthen}{\rightarrow}

\newcommand{\rank}{\mathit{rank}}
\newcommand{\Set}{\mathsf{Set}}
\newcommand{\M}{\hearts}

\newcommand{\cset}[1]{{\{ #1 \}}}
\newcommand{\tup}[1]{{\langle #1 \rangle}}
\newcommand{\bbrack}[1]{{\llbracket #1 \rrbracket}}

\newcommand{\fP}{\mathcal{P}}
\newcommand{\fQ}{\mathcal{Q}}
\newcommand{\natto}{\dot\to}

\makeatletter
\newcommand*{\@old@slash}{}\let\@old@slash\slash
\def\slash{\relax\ifmmode\delimiter"502F30E\mathopen{}\else\@old@slash\fi}
\makeatother

\spnewtheorem{thm}[theorem]{Theorem}{\bfseries}{\itshape}
\spnewtheorem{cor}[theorem]{Corollary}{\bfseries}{\itshape}
\spnewtheorem{lem}[theorem]{Lemma}{\bfseries}{\itshape}
\spnewtheorem{lemdefn}[theorem]{Lemma and Definition}{\bfseries}{\itshape}
\spnewtheorem{propn}[theorem]{Proposition}{\bfseries}{\itshape}
\spnewtheorem{defn}[theorem]{Definition}{\bfseries}{\upshape}
\spnewtheorem{obs}[theorem]{Observation}{\bfseries}{\upshape}
\spnewtheorem{rem}[theorem]{Remark}{\bfseries}{\upshape}
\spnewtheorem{expl}[theorem]{Example}{\bfseries}{\upshape}
\spnewtheorem{thmdefn}[theorem]{Theorem and Definition}{\bfseries}{\itshape}
\spnewtheorem{propdefn}[theorem]{Proposition and Definition}{\bfseries}{\itshape}
\spnewtheorem{assn}[theorem]{Assumption}{\bfseries}{\upshape}
\spnewtheorem{algorithm}[theorem]{Algorithm}{\bfseries}{\upshape}
\spnewtheorem{conv}[theorem]{Convention}{\bfseries}{\upshape}
\spnewtheorem{notn}[theorem]{Notation}{\bfseries}{\upshape}
\spnewtheorem{open}[theorem]{Problem}{\bfseries}{\upshape}

\newcommand{\into}{\hookrightarrow}

\newcommand{\fD}{\+D} 
\newcommand{\fB}{\+B_{\infty}} 
\newcommand{\fmN}{\+M}

\begin{document}

\title{Simulations and Bisimulations For Coalgebraic Modal Logics}
\author{
 Daniel Gor{\'\i}n and Lutz Schr{\"o}der
}
\institute{Department of Computer Science, Universit{\"a}t Erlangen-N{\"u}rnberg}
\maketitle

\begin{abstract}
  We define a notion of $\Lambda$-simulation for coalgebraic modal
  logics, parametric on the choice $\Lambda$ of predicate liftings for
  a functor $T$.  We show this notion is adequate in several ways: i)
  it preserves truth of positive formulas, ii) for $\Lambda$ a
  separating set of monotone predicate liftings, the associated notion
  of $\Lambda$-bisimulation corresponds to $T$-behavioural equivalence
  (moreover $\Lambda$-$n$-bisimulations correspond to
  $T$-$n$-behavioural equivalence), and iii) in fact, for
  $\Lambda$-separating and $T$ preserving weak pullbacks, difunctional
  $\Lambda$-bisimulations are $T$-bisimulations. In essence, we arrive
  at a modular notion of equivalence that, when used with a separating
  set of monotone predicate liftings, coincides with $T$-behavioural
  equivalence regardless of whether $T$ preserves weak pullbacks
  (unlike the notion of $T$-bisimilarity).
\end{abstract}

\section{Introduction}

As the basic notion of equivalence in coalgebra, $T$-behavioural
equivalence has emerged, which declares two states to be equivalent if
they are identified by some pair of coalgebra morphisms; in case the
type functor $T$ admits a final coalgebra, $T$-behavioural equivalence
is just identification in the final $T$-coalgebra. As a proof
principle, however, $T$-behavioural equivalence is comparatively
unwieldy, thus motivating the search for bisimulation-type proof
principles whereby two states can be shown to be behaviourally
equivalent by exhibiting a \emph{bisimulation relation} between
them. The advantage of such approaches is that bisimulation relations
may be comparatively small, making equivalence proofs by bisimulation
more manageable than direct proofs of behavioural equivalence.

The downside is that while behavioural equivalence is a canonical
notion that works for any type of coalgebras, it is rather less clear
what a bisimulation is in general. In case the type functor preserves
weak pullbacks, the standard notion of \emph{$T$-bisimulation} gives a
satisfactory answer: it can be uniformly defined for any $T$, it is
always sound for $T$-behavioural equivalence, if $T$ preserves weak
pullbacks it is complete for $T$-behavioural equivalence, and it
coincides with standard notions in the main examples. For functors
that fail to preserve weak pullbacks, however, the search for a good
generic notion of bisimilarity remains largely open.

Here, we present a modally-inspired notion of bisimulation that partly
solves these problems, specifically it does so for functors that admit
a separating set of \emph{monotone} predicate liftings. Our notion of
\emph{$\Lambda$-bisimilarity} depends on distinguishing a modal
signature $\Lambda$ that we assume to consist of monotone
operators. Key features of $\Lambda$-bisimilarity are
\begin{itemize}
\item It is related to a corresponding notion of
  \emph{$\Lambda$-simulation}, which bears a clear relation to modal
  logic: all positive modal formulas over $\Lambda$ are preserved by
  $\Lambda$-simulations.
\item If $\Lambda$ is separating, then $\Lambda$-bisimulation is
  sound and complete for behavioural equivalence.
\item We have a finite-lookahead version of
  $\Lambda$-bisimilarity. This $\Lambda$-$n$-bisimilarity is sound and
  complete for the standard notion of $n$-behavioural equivalence
  defined via the terminal sequence.
\item $\Lambda$-bisimulation allows bisimulation proofs up to
  difunctionality (i.e.\ closure under zig-zags).
\item If $T$ preserves weak pullbacks, then $\Lambda$-bisimulations
  are essentially the same as $T$-bisimulations, at least when we
  restrict to difunctional relations.
\end{itemize}

\paragraph{Related Work:} Recent yet unpublished work by
Enqvist~\cite{Enqvist13} introduces a notion of
\emph{$\Lambda$-homomorphism} that is almost a special case of a
$\Lambda$-simulation, and in fact shows that such
$\Lambda$-homomorphisms can be induced by a relator in the sense
of~\cite{Levy11}, so that the notion of $\Lambda$-simulation can
itself be regarded as implicit in that work. When we say `almost', we
mean that the implication in the definition of $\Lambda$-homomorphism
goes the other way in Enqvist's work than it does here, so that in
particular Theorem~\ref{thm:sim-preserves-truth} would fail for his
notion. The notion of $\Lambda$-homomorphism in the version that
appears here has been under discussion between the authors' group and
international coauthors from late 2011.

In~\cite{MartiVenema12} it is shown that so-called \emph{lax
  extensions of $T$ preserving diagonals} induce notions of
bisimulation that are sound and complete for behavioural equivalence,
and that a finitary functor has such an extension iff it admits a
separating set of finitary monotone predicate liftings. Our result,
while otherwise working with similar assumptions, does not suppose
finitaryness of the functor.

In~\cite{Levy11} a generic theory of coalgebraic simulation is
developed using \emph{relators}. One can show that our notion of
$\Lambda$-simulation is induced by a relator and therefore subsumed by
that framework. We cannot currently make out that any of our results
about $\Lambda$-(bi)simulation could be obtained by instantiating the
generic results, however.

\section{Preliminaries}\label{sec:prelim}

The framework of \emph{coalgebraic modal logic}~\cite{Pattinson03}
covers a broad range of modalities beyond the standard relational setup,
including probabilistic and game-theoretic phenomena as well as neighbourhood
semantics and non-material conditionals~\cite{SchroderPattinson09a}.
This framework is parametric in syntax and semantics. The syntax is
given by a \emph{similarity type} $\Lambda$, i.e.\ a set of
\emph{modal operators} with finite arities $\ge 0$ (hence possibly
including propositional atoms). To simplify notation, we will pretend
that all operators are unary.

\begin{defn}\label{def:formulas}
  The set $L(\Lambda)$ of \emph{$\Lambda$-formulas} is given by the grammar:
  \begin{equation*}
    \phi,\psi::=\top\mid \lnot\phi \mid \phi\land\psi\mid\M\phi
    \qquad(\M\in\Lambda).
  \end{equation*}
\end{defn}

\noindent
We use the standard derived Boolean operators
$\lor$, $\lthen$, etc.
We use $\rank(\phi)$ to denote the maximum number of nested occurrences of
$\M \in \Lambda$ in $\phi$.

Semantics are parametrized by associating a $\Lambda$-structure
$\tup{T,\cset{\bbrack{\M_\lambda}}_{\lambda\in \Lambda}}$ to
a similarity type $\Lambda$. Here $T$ is an endofunctor $T$ on
the category $\Set$ and, each $\bbrack{\M_\lambda}$ is a
\emph{predicate lifting}, that is, a natural transformation
$\bbrack{\M} : \fQ \natto
\fQ\circ T^\mathit{op}$, where $\fQ$ is the contravariant powerset
functor $\Set^\mathit{op} \to \Set$ (that is, $\fQ X \mapsto 2^X$ for every set $X$,
and given $f:X\to Y$, $\fQ f : 2^Y \to 2^Y$ is given by
$\fQ f \mapsto \lambda A . f^{-1}[A]$). For the extension of predicate liftings
to the higher-arity case see~\cite{explimits}.

\begin{assn}\label{ass:inj}
  We can assume w.l.o.g.\ that $T$ preserves injective
  maps~\cite{Barr93}.  For convenience of notation, we will in fact
  sometimes assume that subset inclusions $X \into Y$ are mapped
  to subset inclusions $TX \into TY$. Moreover, we assume
  w.l.o.g.\ that $T$ is non-trivial, i.e.\ $TX=\emptyset\implies
  X=\emptyset$ (otherwise, $TX=\emptyset$ for all $X$).
\end{assn}

\noindent
We typically identify a similarity type $\Lambda$ and its associated
$\Lambda$-structure, and refer to both as $\Lambda$. Unless otherwise
stated, $T$ stands for the underlying functor of the given $\Lambda$-structure.

For a given choice of $\Lambda$, a model for $L(\Lambda)$ is just a
\emph{$T$-coalgebra} $\tup{X,\xi}$, i.e.\ a non-empty set $X$ (the set of
\emph{states}) and \emph{transition function} $\gamma : X \to TX$.
Given $x \in X$, the truth value of $L(\Lambda)$-formulas is defined as:
\begin{align}
 x \models_\gamma \top &\phantom{\iff} \text{ always}
 \\
 x \models_\gamma \lnot\phi  &\iff x \not\models_\gamma \phi
 \\
 x \models_\gamma \phi \land \psi &\iff
   x \models_\gamma \phi \text{ and } x \models_\gamma \psi
 \\
 x \models_\gamma \M\phi &\iff \gamma(x) \models \M\bbrack{\phi}_\gamma
\end{align}
 where $\bbrack{\phi}_\gamma$, the extension of $\phi$ in $\gamma$ is given by
 $\bbrack{\phi}_\gamma = \cset{z \in X \mid x \models_\gamma \phi}$.
 and for $t \in TX$ and $A \subseteq X$, $t \models \M A$
 is a more suggestive notation for $t \in \bbrack{\M}_X A$.
When clear from context, we shall write simply $x \models \phi$ and
$\bbrack{\phi}$.

\begin{expl}\label{ex:K}
  Coalgebras for the (covariant) finite powerset functor $\fP_\omega$ are
  finitely branching directed graphs. For a similarity type
  $\Lambda=\cset{\Box,\Diamond}$
  consider the associated predicate liftings:
    \begin{align}
      \bbrack{\Box}_X(A) &:= \cset{B \mid B \subseteq A}
      \\
      \bbrack{\Diamond}_X(A) &:= \cset{B \mid B \cap A \neq \emptyset}
    \end{align}
    They correspond to the classical modal operators of relational
    modal logics, so the logic we get in this case is essentially the
    mono-modal version of the Hennessy-Milner logic~\cite{HM80}. To
    obtain the basic modal logic $K$ one needs to enrich the coalgebra
    structure with an interpretation for propositions.  So let $V$ be
    a set of proposition symbols and let $C_V$ be the constant functor
    that maps every set $X$ to $2^V$. For each $p \in V$, the
    (nullary) predicate lifting $\bbrack{p}_X := \cset{\pi \in 2^V
      \mid p \in \pi}$ describes structures satisfying $p$. The Kripke
    functor $K$ is then defined as $KX := C_V \times \fP X$ and the
    similarity type $\Lambda = V \cup \cset{\Diamond,\Box}$ is
    interpreted using the corresponding predicate liftings on the
    appropriate projections.
\end{expl}

\begin{expl}
 The language of \emph{graded modal logic} corresponds to the set
 $\Lambda = \cset{\Diamond_k \mid k \in \mathbb{N}}$ and is interpreted over the
 infinite multiset functor $\fB$, i.e.,
 $\fB X \mapsto \cset{f : X \to \mathbb{N} \cup \cset{\infty} \mid \text{$f$ has finite support}}$.
 Coalgebras for $\fB$ are finitely branching multigraphs (with potentially infinite cardinalities).
 Interpretation of the modal operators is by way of the following family of predicate liftings,
 for each $k \in \mathbb{N}$:
 \begin{equation}\textstyle
   \bbrack{\Diamond_k}_X(A) := \cset{b \in \fB X \mid b(A) > k}\enspace
 \end{equation}
 where by $b(A)$ we denote $\sum_{x \in A}b(x)$, i.e.\ we use $b\in\fB
 X$ like measure on $X$.
\end{expl}

\begin{expl}
 Probabilistic modal logics are obtained when one takes the functor $\fD$
 that maps $X$ to the set of finitely-supported probability distributions over $X$.
 For the language $\Lambda_M=\cset{M_p\mid p\in [0,1]\cap\mathbb{Q}}$, with $M_p$ informally
 read as ``with probability more than $p$'', the
 corresponding predicate liftings are defined analogously as for graded modal logics.
 One can instead take $\Lambda_P=\cset{L_p\mid p\in [0,1]\cap\mathbb{Q}}$, with
 $L_p$ read as ``with probability at least $p$'', and interpreted using:
   \begin{equation}\textstyle
     \bbrack{L_p}_X(A) := \cset{\mu \in \fD X \mid \mu(A) \ge k}\enspace .
   \end{equation}
\end{expl}

\begin{expl}
 As a final example, consider the subfunctor $\mathcal{M}$ of $\fQ \circ \fQ$ given by
 $\mathcal{M} X = \cset{S \in \fQ\fQ X| \text{$S$ is upwards closed}}$. Over this functor
 one can obtain the monotone neighborhood semantics of modal logic with $\Lambda = \cset{\Box}$
 using the predicate lifting $\bbrack{\Box}_X(A) := \cset{S \in \mathcal{M}X \mid A \in S}$.
\end{expl}

\noindent
A modal operator $\M$ is called \emph{monotone} if it satisfies the condition
\begin{equation}\label{eq:monot}
A \subseteq B \subseteq X \text{ implies } \bbrack{\M}_X A \subseteq \bbrack{\M}_X B\enspace .
\end{equation}
While all the examples above are monotone, it is worth stressing that
the framework of coalgebraic modal logics can indeed accommodate non-monotone logics.
We will however focus on the monotone case.

\begin{assn}
In the following, we assume all modal operators to be monotone.
\end{assn}

\noindent
For a given endofunctor $T$, the choice of both the similarity type
$\Lambda$ and the associated $\Lambda$-structure over $T$ may vary
(although the number of choices is formally limited~\cite{explimits}),
and each choice yields a potentially different logic. When the choice
of predicates of liftings in $\Lambda$ is rich enough as to uniquely
describe every element in $TX$, we call such $\Lambda$
\emph{separating}~\cite{expterminal}:
\begin{defn}
We say that $\Lambda$ is \emph{separating} if $t \in TX$ is uniquely
determined by the set
$\cset{(\M, A) \in \Lambda \times \fP X \mid t \models \M A}$.
\end{defn}

\noindent It is not hard to see that, for example, $V \cup
\cset{\Box}$ as well as $V \cup \cset{\Diamond}$ are separating over
the Kripke functor $K$ of Example~\ref{ex:K}. The reader is referred
to~\cite{explimits} for characterizations of functors that admit
separating sets of predicate liftings.

\begin{defn}
  Given $T$-coalgebras $C$ and $D$, we say that states $x$ in $C$ and
  $y$ in $D$ are \emph{behaviourally equivalent}, and write
  $(C,x)\approx(D,y)$, or shortly $x\approx y$, whenever there exists
  a $T-$coalgebra $E$ and coalgebra morphisms $f : C \to E$ and $g : D
  \to E$ such that $f(x) = g(y)$.
\end{defn}
\noindent Simulations like the ones we will present in
Section~\ref{sec:sim} occur frequently when dealing with logics that
do not contain a Boolean basis; typically, negation is absent or only
allowed on restricted positions (e.g., in front of atoms).  The notion
of \emph{positive} formula is a generalization of this idea.

\begin{defn}
  The language $L^+(\Gamma)$ of \emph{positive $\Lambda$-formulas} is given by:
  \begin{equation*}
    \phi,\psi::=\top\mid\bot\mid \phi\land\psi\mid\phi\lor\psi\mid\M\phi
    \qquad(\M\in\Lambda).
  \end{equation*}
\end{defn}

\noindent
We can regard $L^+(\Lambda)$ as a syntactic fragment of $L(\Lambda)$
where $\lor$ is now taken as primitive.  The Boolean connectives of
$L^+(\Lambda)$ allow expressing all the monotone Boolean functions,
but notice that $\Lambda$ may contain dual operators (e.g.,
$\Lambda=\cset{\Box,\Diamond}$) --- in fact if $\Lambda$ is closed
under dual operators then $L^+(\Lambda)$ is as expressive as
$L(\Lambda)$. In general, of course, $L^+(\Lambda)$ is a proper
fragment of $L(\Lambda)$.

\section{Coalgebraic simulation}\label{sec:sim}

We now proceed to introduce our notion of modal simulation. We use
standard notation for relations; in particular, given a binary
relation $S\subseteq X\times Y$ and $A\subseteq X$, we denote by
$S[A]$ the relational image $S[A]=\{y\mid\exists x\in A.\,xSy\}$.

\begin{defn}[$\Lambda$-Simulation, $\Lambda$-Homomorphism]
  Let $C=(X,\xi)$ and $D=(Y,\zeta)$ be $T$-coalgebras. A
  \emph{$\Lambda$-simulation} $S:C\to D$ (\emph{of $D$ by $C$}) is a
  relation $S\subseteq X\times Y$ such that whenever $xSy$ then for
  all $\M\in\Lambda$ and all $A\subseteq X$
  \begin{equation*}
    \xi(x)\models \M A\text{ implies }\zeta(y)\models
    \M S[A].
  \end{equation*}
  A function $f:X\to Y$ is a \emph{$\Lambda$-homomorphism} if its
  graph is a $\Lambda$-simulation.
\end{defn}

\begin{lem}\label{lem:stability}
  $\Lambda$-simulations are stable under unions and relational
  composition. Moreover, equality is always a $\Lambda$-simulation.
\end{lem}

\begin{defn}[$\Lambda$-ordering]
  The \emph{$\Lambda$-preorder} $\le_\Lambda$ on $TX$ is defined by
  \begin{equation*}
    s\le_\Lambda t\iff \forall\M\in\Lambda,A\subseteq X.
    (s\models\M A\impl t\models\M A).
  \end{equation*}
\end{defn}

\begin{lem}\label{lem:hom}
  Let $C=(X,\xi)$ and $D=(Y,\zeta)$ be $T$-coalgebras. A map $f:X\to
  Y$ is a $\Lambda$-homomorphism iff for all $x\in Y$,
  \begin{equation}\label{eq:hom}
    Tf(\xi(x))\le_\Lambda\zeta(f(x)).
  \end{equation}
\end{lem}
\begin{proof}
  \emph{`Only if':} Let $\M\in\Lambda$, $A\subseteq Y$. Then
  \begin{align*}
    Tf(\xi(x))\models\M A  \iff& \xi(x)\models\M f^{-1}[A] &&\text{(naturality)}\\
     \impl & \zeta(f(x))\models\M f[f^{-1}[A]] && \text{(simulation)}\\
     \impl & \zeta(f(x))\models\M A && \text{(monotony)}.
  \end{align*}

  \emph{`If':} Let $\xi(x)\models\M A$. We have to show
  $\zeta(f(x))\models\M f[A]$, which will follow by
  (\ref{eq:hom}) from $Tf(\xi(x))\models\M f[A]$. By naturality,
  the latter is equivalent to $\xi(x)\models\M
  f^{-1}[f[A]]$. This however follows from $\xi(x)\models\M A$ by
  monotony.\qed
\end{proof}
\begin{rem}
  In the notation of the above lemma, another equivalent formulation
  of $f$ being a $\Lambda$-homomorphism is that $\xi(x)\models\M
  f^{-1}[A]$ implies $\zeta(f(x))\models\M A$ for $\M\in\Lambda$,
  $A\subseteq Y$. This is an immediate consequence of the lemma by
  naturality of predicate liftings.
\end{rem}
As announced, $\Lambda$-simulations preserve the truth of positive
modal formulas over $\Lambda$:
\begin{thm}\label{thm:sim-preserves-truth}
  If $S$ is a simulation and $xSy$, then $x\models\phi$ implies
  $y\models\phi$ for every positive $\Lambda$-formula $\phi$.
\end{thm}
\begin{proof}
  Induction over $\phi$, with trivial Boolean cases (noting that these
  do not include negation). For the modal case, we have
  \begin{align*}
    x\models\M\phi & \iff \xi(x)\models \M\Sem{\phi}\\
    & \impl \zeta(y)\models\M\{\{y'\mid
    \exists x'.(x'\models\phi\land x'Sy')\}\\
    & \impl \zeta(y)\models\M\Sem{\phi}\\
    & \iff y\models\M\phi.
  \end{align*}
  \qed
\end{proof}

\begin{expl}
  \begin{enumerate}
  \item When $\Lambda=\{\Diamond\}$, then a $\Lambda$-simulation
    $S:C\to D$ is just a simulation $C\to D$ in the usual
    sense. (Proof: `only if': if $xSy$ and $x'\in\xi(x)$, then
    $\xi(x)\models\Diamond\{x'\}$ and hence
    $\zeta(y)\models\Diamond\{y'\mid x'Sy'\}$, i.e.\ there exists $y'$
    such that $x'Sy'$ and $y'\in\zeta(y)$. `If': If
    $\xi(x)\models\Diamond A$, then there exists $x'\in A\cap\xi(x)$
    and hence we have $y'\in\zeta(y)$ such that $x'Sy'$, so that
    $\zeta(y)\models\Diamond\{y''\mid\exists x''\in\xi(x).\,
    x''Sy''\}$.)
  \item When $\Lambda=\{\Box\}$, then a $\Lambda$-simulation $S:C\to
    D$ is just a simulation $D\to C$ in the usual sense. (Proof: `only
    if': Let $xSy$ and $y'\in\zeta(y)$. Assume that we cannot find
    $x'\in\xi(x)$ such that $x'Sy'$; that is,
    $\xi(x)\models\Box\{x'\mid \neg(x'Sy')$. Then by the definition of
    $\Lambda$-simulation, $\zeta(y)\models\Box A$ for an $A$ with
    $y'\notin A$, contradiction. `If': Let $\xi(x)\models\Box A$. To
    show that $\zeta(y)\models\Box\{y'\mid\exists x'\in A.\,x'Sy'\}$,
    let $y'\in\zeta(y)$. By the simulation property, there exists
    $x'\in\xi(x)$ such that $x'Sy'$, and since $\xi(x)\models\Box A$,
    we have $x'\in A$.)
  \item For probabilistic modal logic, with $\Lambda=\{L_p\mid p\in
    [0,1]\cap\mathbb{Q}\}$, a relation $S\subseteq X\times Y$ between
    $\fD$-coalgebras $(X,\xi)$ and $(Y,\zeta)$ is a
    $\Lambda$-simulation iff for all $xSy$ and all $A\subseteq X$,
    \begin{equation*}
      \zeta(y)(S[A])\ge\xi(x)(A)
    \end{equation*}
    (keep in mind that $\xi(x)$ and $\zeta(y)$ are probability
    measures that we can apply to subsets). The same comes out when we
    take $\Lambda=\{M_p\mid p\in [0,1]\cap\mathbb{Q}\}$. Note that
    standardly, probabilistic bisimulations (see the next section for
    the definition of bisimulations) are defined only for the case
    where $S$ is an equivalence relation, in which case the notion
    coincides with the above.
  \item For graded modal logic, with $\Lambda=\{\Diamond_k\mid
    k\in\mathbb{N}\}$, we obtain the same inequality characterizing
    $\Lambda$-simulations as for probabilistic logic (keeping in mind
    that we can see $\xi(x)\in\fB(X)$, $\zeta(y)\in\fB(Y)$ as discrete
    $\mathbb{N}\cup\{\infty\}$-valued measures).
  \item For monotone neighbourhood logic, with $\Lambda=\{\Box\}$, we
    have that a relation $S\subseteq X\times Y$ between
    $\fmN$-coalgebras $(X,\xi)$ and $(Y,\zeta)$ is a
    $\Lambda$-simulation iff for $xSy$, $A\in\xi(x)$ implies
    $S[A]\in\zeta(y)$. This is easily seen to be equivalent to the
    forth condition in the definition of monotone bisimulation,
    attributed to Pauly in~\cite{HansenKupke04}.
  \end{enumerate}
\end{expl}
\noindent
For many purposes, simulations can be already too strong, e.g.\ when
we are interested in preservation results for positive formulas up to
a certain modal depth. It is therefore natural to consider
$n$-simulations.

\begin{defn}[$\Lambda$-$n$-simulation]
  Let $C=(X,\xi)$ and $D=(Y,\zeta)$ be $T$-coalgebras.  We define the
  notion of \emph{$\Lambda$-$n$-simulation} inductively as follows.
  Any $S_0 \subseteq X \times Y$ is a $\Lambda$-0-simulation. A
  relation $S_{n+1}\subseteq X \times Y$ is a
  $\Lambda$-$(n+1)$-simulation if there exists a
  $\Lambda$-$n$-simulation $S_n$ such that $S_{n+1} \subseteq S_n$ and
  for all $x,y$, $xS_{n+1}y$ implies that for all $\M\in\Lambda$,
  $A\subseteq X$
  \begin{equation*}
    \xi(x)\models \M A \text{ implies } \zeta(y)\models
    \M S_n[A].
  \end{equation*}
\end{defn}

\begin{thm}
  If $S$ is a $\Lambda$-$n$-simulation and $xSy$, then $x\models\phi$ implies
  $y\models\phi$ for every positive $\Lambda$-formula $\phi$ of rank at most $n$.
\end{thm}
\begin{proof}
  Induction on $n$. The base case $n=0$ is trivial since then $\phi$
  is equivalent to either $\top$ or $\bot$. For $n > 0$, we proceed
  by induction on $\phi$, the interesting case being:
  \begin{align*}
    x\models\M\psi  \iff & \xi(x)\models \M\Sem{\psi}\\
     \impl & \zeta(y)\models\M S_{n-1}[\Sem{\psi}]\\
     \impl & \zeta(y)\models\M\Sem{\psi}&&\text{(outer IH $+$ monotony)}\\
     \iff & y\models\M\psi.
  \end{align*}\qed
\end{proof}

\section{Bisimulations for all}

The notion of $\Lambda$-($n$)-simulation naturally yields a notion of
bisimulation (i.e., simulations in both directions). The yardstick for
any notion of bisimulation is \emph{$T$-behavioural equivalence} (see
Section~\ref{sec:prelim}). We say that a notion of bisimulation is
\emph{sound for $T$-behavioural equivalence} if any two states related
by bisimulation are $T$-behaviourally equivalent, and \emph{complete
for $T$-behavioural equivalence} if any two $T$-behaviourally
equivalent states can be related by a bisimulation.

The standard coalgebraic notion of \emph{$T$-bisimulation} that we
recall below is always sound for $T$-behavioural equivalence, and
complete for $T$-behavioural equivalence if $T$ preserves weak pullbacks.
We will show that our notion of $\Lambda$-bisimilarity is always sound
and complete for $T$-behavioural equivalence, provided that $\Lambda$ is
separating.
Notice also that $\Lambda$-bisimulations enjoy nice closure properties,
in particular under unions and composition, which for $T$-bisimulations
is only the case, again, when $T$ preserves weak pullbacks.

  \begin{defn}
    If $S$ and its converse $S^{-1}$ are $\Lambda$-$n$-simulations,
    then $S$ is a \emph{$\Lambda$-$n$-bisimulation}. Analogously, a
    \emph{$\Lambda$-bisimulation} is a $\Lambda$-simulation $S$ such
    that $S^{-1}$ is a $\Lambda$-simulation as well.
  \end{defn}
  \begin{lem}\label{lem:mor}
    If $C$, $D$ are $T$-coalgebras and $f:C\to D$ is a coalgebra
    morphism, then the graph of $f$ is a $\Lambda$-bisimulation.
  \end{lem}
  \begin{proof}
    It follows from Lemma~\ref{lem:hom} that the graph of $f$ is a
    $\Lambda$-simulation. To see that its converse is a
    $\Lambda$-simulation, let $C=(X,\xi)$, $D=(Y,\zeta)$, and let
    $x\in X$, $\M\in\Lambda$, $A\subseteq Y$ such that
    $\zeta(f(x))\models\M A$. Now $\zeta(f(x))=Tf(\xi(x))$ because $f$
    is a coalgebra morphism, so we obtain $\xi(x)\models\M f^{-1}[A]$
    by naturality of predicate liftings, as required. \qed
  \end{proof}

  \noindent
  It is easy to see that $\Lambda$-$n$-bisimulations preserve and
  reflect the truth of formulas with up to $n$ nested modalities. A
  similar notion of preservation, \emph{$n$-step-equivalence} was
  considered in~\cite{SchroderPattinson10b}, obtained by projecting
  into the terminal sequence. We can show that $n$-step-equivalence
  coincides with $\Lambda$-$n$-bisimilarity when $\Lambda$ is
  separating.

\begin{defn}
  The \emph{terminal sequence} of a given functor $T$ is the sequence
  given by $T_0 = 1$ (some singleton set) and $T_{n+1} = TT_n$,
  connected by functions $p_n : T_{n+1} \to T_n$, where $p_{n+1} =
  Tp_n$. Every $T$-coalgebra $C=(X,\xi)$ defines a cone over the
  terminal sequence by $\xi_0 :C \to 1$ (uniquely defined) and
  $\xi_{n+1} = T\xi_n \circ \xi$. Given $T$-coalgebras
  $(X,\xi)$ and $(Y,\zeta)$ and elements $x \in X$, $y \in Y$, we say
  that $x$ and $y$ are \emph{$n$-step equivalent} (notation: $x
  \approx_n y$) whenever $\xi_n(x) = \zeta_n(y)$.
\end{defn}

\begin{lem}
  Let $C=(X,\xi)$ and $D=(Y,\zeta)$ be $T$-coalgebras. The
  $n$-step-equivalence relation $\approx_n \subseteq X \times Y$ is a
  $\Lambda$-$n$-bisimulation.
\end{lem}
\begin{proof}
  Of course, it suffices to show that $\approx_n$ is a
  $\Lambda$-$n$-simulation. We proceed by induction on $n$. Clearly,
  $\approx_0 = X \times Y$ is a $\Lambda$-$0$-simulation.  For the
  inductive step, let $x \approx_{n+1} y$ and let $\M\in\Lambda$, $A
  \subseteq X$ such that $\xi(x) \in \M_C A$. We then have (writing
  $\fP$ and $\fQ$ for the covariant and contravariant powerset
  functors, respectively):

 \begin{align*}
 \xi(x) \models \M A
   &\impl \xi(x) \in \Sem{\M}_C\circ\fQ\xi_n\circ \fP\xi_n A
     &&\text{(monotony)}
     \\
   &\impl \xi(x) \in \fQ(T\xi_n)\circ\Sem{\M}_{T_n}\circ\fP\xi_n A
     &&\text{(naturality)}
     \\
   &\impl x \in \fQ\xi \circ \fQ(T\xi_n)\circ \Sem{\M}_{T_n}\circ \fP\xi_n A
     \\
   &\impl x \in \fQ\xi_{n+1} \circ \Sem{\M}_{T_n} \circ \fP\xi_n A
     &&\text{(functoriality)}
     \\
   &\impl y \in \fQ\zeta_{n+1} \circ \fP\xi_{n+1}\circ \fQ\xi_{n+1}\circ\Sem{\M}_{T_n} \circ \fP\xi_nA
     &&\text{($x \approx_{n+1} y$)}
     \\
   &\impl y \in \fQ\zeta_{n+1}\circ\Sem{\M}_{T_n}\circ\fP\xi_n A
     &&\text{($(\fP f \circ \fQ f) X \subseteq X$)}
     \\
   & \qquad = \fQ\zeta \circ \fQ(T\zeta_n) \circ \Sem{\M}_{T_n} \circ \fP\xi_n A
     \\
   &\impl \zeta(y) \in \fQ(T\zeta_n)\circ\Sem{\M}_{T_n}\circ\fP\xi_n A
     \\
   &\impl \zeta(y) \in \Sem{\M}_D\circ \fQ\zeta_n\circ\fP\xi_n A
     &&\text{(naturality)}
     \\
   &\impl \zeta(y) \models \M\approx_n\!\![A].
 \end{align*}
 By the inductive hypothesis, $\approx_n$ is a
 $\Lambda$-$n$-simulation, and, moreover, $\approx_{n+1} \subseteq
 \approx_n$, so $\approx_{n+1}$ is a
 $\Lambda$-$(n+1)$-simulation. \qed
\end{proof}
\noindent
Of course, the converse of this lemma does not hold in general (e.g.,
take $T$ to be the multiset functor and consider $\Lambda =
\cset{\Diamond_0}$). However, we do have the following.

\begin{thm}\label{thm:n-bisim-n-step-equiv}
If $\Lambda$ is a separating set of predicate liftings, then $S_n \subseteq {\approx_n}$ for every
$\Lambda$-$n$-bisimulation $S_n$.
\end{thm}
\begin{proof}
  Induction on $n$.  Let $C=(X,\xi)$, $D=(Y,\zeta)$ be $T$-coalgebras,
  let $S_{n+1}\subseteq X\times Y$ be a
  $\Lambda$-$(n+1)$-bisimulation, and let $x S_{n+1} y$. Let
  $S_n\supseteq S_{n+1}$ be an $n$-bisimulation as in the definition
  of $\Lambda$-$(n+1)$-bisimilarity.

  We show $\xi_{n+1}(x)=\zeta_{n+1}(y)$ using separation. Thus, let
  $\M \in \Lambda$, $A \subseteq T_n$. We have to show that
  $\xi_{n+1}(x) \models \M A$ iff $\zeta_{n+1}(y)\models\M A$; by
  symmetry, it suffices to prove `only if'.  Since
  $\xi_{n+1}=T\xi\xi_n$, we have, by naturality, $\xi(x)\models
  \M\xi_n^{-1}[A]$. By simulation, it follows that $\zeta(y)\models \M
  S_n[\xi_n^{-1}[A]]$. By the inductive hypothesis, $S_n \subseteq
  \approx_n$, so that we obtain $\zeta(y)\models \M
  \approx_n\!\![\xi_n^{-1}[A]]$ by monotony. Now
  $\approx_n\!\![\xi_n^{-1}[A]]=\zeta_n^{-1}[A]$ by definition of
  $\approx_n$, and hence $\zeta_{n+1}(y)=T\zeta_n\zeta(y)\models\M A$
  by naturality. \qed
\end{proof}
\noindent In other words, $\Lambda$-$n$-bisimulation is always
complete for $n$-step equivalence, and sound if $\Lambda$ is
separating.

Similar results hold for $\Lambda$-bisimulations. Specifically, we
have
\begin{lem}\label{lem:complete}
  The behavioural equivalence relation $\approx$ between two given
  $T$-coalgebras is a $\Lambda$-bisimulation.
\end{lem}
\noindent In other words, $\Lambda$-bisimulation is always complete
for behavioural equivalence.
\begin{proof}
  Let $C=(X,\xi)$, $D=(Y,\zeta)$ be $T$-coalgebras; it suffices
  to show that behavioural equivalence $\approx$ (as a relation between $X$ and $Y$)
  is a $\Lambda$-simulation between $C$ and $D$.
  Given $x \approx y$, $\M\in\Lambda$ and $A\subseteq X$ such that $\xi(x)\models\M A$,
  we then have to show that $\zeta(y)\models\M(\approx\!\![A])$.
  So let $E$ be a $T$-coalgebra and $f:C\to E$ and $g:D\to E$ be coalgebra
  morphisms such that $f(x)=g(y)$.
  By Lemma~\ref{lem:mor}, and by stability of simulations under
  composition, the relation $g^{-1}f=\{(x',y')\mid f(x')=g(y')\}$
  is a $\Lambda$-simulation.  Thus, we have $\zeta(y)\models\M
  g^{-1}[f[A]]$; and because $g^{-1}f$ is contained in $\approx$
  we are done by monotony. \qed
\end{proof}
\noindent As in the bounded-depth setting, soundness depends, of
course, on separation:
\begin{thm}
  If $\Lambda$ is separating, then $\Lambda$-bisimilarity is sound and
  complete for behavioural equivalence.
\end{thm}
\begin{proof}
  As stated above, Lemma~\ref{lem:complete} proves completeness; it
  remains to show soundness. Let $C=(X,\xi)$ and $D=(Y,\zeta)$ be
  $T$-coalgebras, and let $S\subseteq X\times Y$ be a
  $\Lambda$-bisimulation. Let $Z$ be the quotient of the disjoint sum
  $X+Y$ by the equivalence relation generated by $S$, and let
  $\kappa_1:X\to Z$ and $\kappa_2:Y\to Z$ denote the prolongations of
  the coproduct injections into the quotient. It suffices to define a
  coalgebra structure $\chi$ on $Z$ that makes $\kappa_1$ and $\kappa_2$ into
  coalgebra morphisms. We thus have to show that putting
  \begin{align*}
    \chi(\kappa_1(x)) & = T\kappa_1(\xi(x))\\
    \chi(\kappa_2(y)) & = T\kappa_2(\zeta(x))
  \end{align*}
  yields a well-defined map $Z\to TZ$. To this end, it suffices to
  show that $T\kappa_1(\xi(x))= T\kappa_2(\zeta(y))$ whenever
  $xSy$. We prove this using separation by showing that
  $T\kappa_1(\xi(x))\models\M A$ iff $T\kappa_2(\zeta(y))\models \M A$
  for $\M\in\Lambda$, $A\subseteq Z$. We prove only the left-to-right
  implication, the converse one being symmetric. So let
  $T\kappa_1(\xi(x))\models\M A$. Then $\xi(x)\models\M\kappa_1^{-1}[A]$
  by naturality, and hence $\zeta(y)\models\M S[\kappa_1^{-1}[A]]$
  since $S$ is a $\Lambda$-simulation. Now clearly
  $S[\kappa_1^{-1}[A]]\subseteq\kappa_2^{-1}[A]$, so that
  $\zeta(y)\models\M\kappa_2^{-1}[A]$ by monotony. We are done by
  naturality. \qed

\end{proof}
In the case where $T$ preserves weak pullbacks, it is well-known that
$T$-bisimilarity in the sense of Aczel and Mendler is also sound and
complete for behavioural equivalence, so that $T$-bisimilarity and
$\Lambda$-bisimilarity coincide when $\Lambda$ is separating. But we
can do better: $T$-bisimulations are $\Lambda$-bisimulations (so
$\Lambda$-simulations are at least as convenient a tool as
$T$-bisimulations), and for $T$ preserving weak pullbacks and
$\Lambda$ separating, difunctional $\Lambda$-bisimulations are
$T$-bisimulations. We recall the relevant definitions:
\begin{defn}
  A \emph{$T$-bisimulation} between $T$-coalgebras $(X,\xi)$ and
  $(Y,\zeta)$ is a relation $S\subseteq X\times Y$ such that there
  exists a coalgebra structure $\rho:S\to TS$ that makes the
  projections $S\to X$ and $S\to Y$ into coalgebra morphisms.
\end{defn}

\begin{defn}
  A binary relation $S\subseteq X\times Y$ is \emph{difunctional} if
  whenever $xSy$, $zSy$, and $zSw$, then $xSw$.
\end{defn}
\noindent Essentially, we obtain a difunctional relation if we take an
equivalence relation $S$ on the disjoint union $X+Y$ of two sets and
restrict it to $X\times Y$, i.e.\ take $S\cap (X\times Y)$ (where
originally $S\subseteq(X+Y)\times(X+Y)$).

We now prove that all $T$-bisimulations are $\Lambda$-bisimulations,
for any $\Lambda$ and $T$, and that the converse holds for
difunctional relations if $T$ preserves weak pullbacks. We conjecture
that the assumption of difunctionality can actually be
removed. Nevertheless, we note the following. To begin, every relation
$S\subseteq X\times Y$ has a difunctional closure $\bar S$, where
$x\bar S y$ iff there exists chains $x=x_0,\dots,x_n$ in $X$ and
$y_0,\dots,y_n=y$ in $Y$ such that $x_iSy_i$ for $i=0,\dots,n$ and
$x_{i+1}Sy_i$ for $i=0,\dots,n-1$.
\begin{defn}
  A \emph{$\Lambda$-bisimulation up to difunctionality} between
  $T$-coalgebras $(X,\xi)$ and $(Y,\zeta)$ is a relation $S\subseteq
  X\times Y$ such that whenever $xSy$ and $\xi(x)\models\M A$ for
  $\M\in\Lambda$, $A\subseteq X$, then $\zeta(y)\models\M\bar S[A]$,
  where $\bar S$ denotes the difunctional closure of $S$, and the
  analogous condition holds for $S^{-1}$.
\end{defn}
\begin{propn}\label{prop:difunctional}
  Let $S\subseteq X\times Y$ be a relation between $T$-coalgebras
  $(X,\xi)$ and $(Y,\zeta)$. Then $S$ is a $\Lambda$-bisimulation up
  to difunctionality iff the difunctional closure of $S$ is a
  $\Lambda$-bisimulation.
\end{propn}
\begin{proof}
  `If' is trivial; we show `only if'. Let $\bar S$ be the difunctional
  closure of $S$. Let $\M\in\Lambda$, $A\subseteq X$ such that $\xi(x)
  \models \M A$, and let $x\bar S y$, i.e.\ we have
  $x=x_0,\dots,x_n\in X$ and $y_0,\dots,y_n=y\in Y$ such that
  $x_iSy_i$ for $i=0,\dots,n$ and $x_{i+1}Sy_i$ for
  $i=0,\dots,n-1$. We define $A_0,\dots,A_n\subseteq X$ and
  $B_0,\dots,B_n\subseteq Y$ inductively by $A_0=A$, $B_i=\bar
  S[A_i]$, and $A_{i+1}=\bar S^{-1}[B_i]$. By induction,
  $\xi(x_i)\models\M A_i$ and $\zeta(y_i)\models\M B_i$ for all
  $i$. Moreover, by difunctionality of $\bar S$, $B_i= \bar S[A]$ for
  all $i$, so that $\zeta(y)=\zeta(y_n)\models\M\bar S[A]$ as
  required. The proof that $\bar S^{-1}$ is also a
  $\Lambda$-simulation is completely analogous.\qed
\end{proof}
\begin{cor}
  Let $\Lambda$ be separating. Then $\Lambda$-bisimilarity up to
  difunctionality is sound and complete for $T$-behavioural
  equivalence.
\end{cor}
\noindent To complement this, we explicitly define a notion of
$T$-bisimulation up to difunctionality:
\begin{defn}
  A \emph{$T$-bisimulation up to difunctionality} between
  $T$-coalgebras $(X,\xi)$ and $(Y,\zeta)$ is a relation $S\subseteq
  X\times Y$ such that there exists a map $\rho:S\to T\bar S$, where
  $\bar S$ denotes the difunctional closure of $S$, such that $T\bar
  p_1\rho=\xi p_1$ and $T\bar p_2\rho=\zeta p_2$. Here $p_1:S\to X$,
  $p_2:S\to Y$, $\bar p_1:\bar S\to X$, and $\bar p_2:\bar S\to Y$
  denote the projections.
\end{defn}
\noindent It does not seem clear in general that an analogue of
Proposition~\ref{prop:difunctional} holds for $T$-bisimulations. For
the case where $T$ preserves weak pullbacks, such an analogue will
follow from the identification with $\Lambda$-bisimulations.

\begin{thm}
  Every $T$-bisimulation (up to difunctionality) is a
  $\Lambda$-bisimulation (up to difunctionality).
\end{thm}
\begin{proof}
  Let $(X,\xi)$ and $(Y,\zeta)$ be $T$-coalgebras. For the plain case,
  let $S\subseteq X\times Y$ be a $T$-bisimulation between them. Thus,
  we have $\rho:S\to TS$ such that $p_1:S\to X$ and $p_2:S\to Y$ are
  coalgebra morphisms. Now let $\M\in\Lambda$, $A\subseteq X$, and
  $xSy$ such that $\xi(x)\models\M A$. We have to show
  $\zeta(y)\models\M S[A]$. Now $\xi(x)= Tp_1\rho(x,y)$, and hence
  $\rho(x,y)\models\M p_1^{-1}[A]$. Since $\zeta(y)=Tp_2\rho(x,y)$, we
  have to show $\rho(x,y)\models\M p_2^{-1} S[A]$. By monotonicity, it
  suffices to show that $p_1^{-1}[A]\subseteq p_2^{-1} S[A]$. So let
  $(x',y')\in S$ such that $x'\in A$; we have to show $y'\in S[A]$, which
  holds by definition of $S[A]$.

  For the second part, let $S$ be a $T$-bisimulation up to
  difunctionality between $(X,\xi)$ and $(Y,\zeta)$, and let $\bar S$
  denote the difunctional closure of $S$. Thus, we have $\rho:S\to
  T\bar S$ such that $T\bar p_1\rho=\xi p_1$ and $T\bar p_2\rho=\zeta
  p_2$, where $p_1:S\to X$, $p_2:S\to Y$, $\bar p_1:\bar S\to X$,
  $\bar p_2:\bar S\to Y$ denote the projections. Let $\M\in\Lambda$,
  $A\subseteq X$ such that $\xi(x)\models\M A$; we have to show
  $\zeta(y)\models\M\bar S[A]$. As above, we find that we equivalently
  need to show $\rho(x,y)\models\M \bar p_2^{-1}[\bar S[A]]$ from
  $\rho(x,y)\models\M \bar p_1^{-1}[A]$, which follows from $\bar
  p_1^{-1}[A]\subseteq \bar p_2^{-1}[\bar S[A]]$.\qed
\end{proof}
\noindent The announced partial converse to this is
\begin{thm}\label{thm:t-bisim}
  If $\Lambda$ is separating and $T$ preserves weak pullbacks, then
  difunctional $\Lambda$-bisimulations are $T$-bisimulations, and
  $\Lambda$-bisimulations up to difunctionality are $T$-bisimulations
  up to difunctionality.
\end{thm}
\noindent 


\begin{proof}
  For the first part, let $S\subseteq X\times Y$ be a difunctional
  $\Lambda$-bisimulation between $T$-coalgebras $(X,\xi)$ and
  $(Y,\zeta)$. Let $p_1:S\to X$ and $p_2:S\to Y$ denote the
  projections. Let
  \begin{equation*}
    \xymatrix{S \ar[r]^{p_1}\ar[d]_{p_2} & X \ar[d]^{q_2} \\
      Y \ar[r]_{q_2} & Z}
  \end{equation*}
  be a pushout; since $S$ is difunctional, this is also a
  pullback. Now observe that the square
  \begin{equation*}
    \xymatrix{S \ar[r]^{p_1}\ar[d]_{p_2} & X \ar[r]^\xi & TX \ar[dd]^{Tq_2} \\
      Y \ar[d]_\zeta \\
    TY \ar[rr]_{Tq_2} & & TZ}
  \end{equation*}
  commutes. To show this, we use separation: let $\M\in\Lambda$, and
  let $A\subseteq Z$. After one application of naturality, we have to
  show that when $xSy$ then $\xi(x)\models\M q_1^{-1}[A]$ iff
  $\zeta(x)\models\M q_2^{-1}[A]$. We show `only if': observe that $Z$
  arises from $X+Y$ by quotienting modulo the equivalence relation
  $\sim_S$ generated by $S$. Thus $q_1^{-1}[A]$ consists of the
  elements of $X$ that are $\sim_S$-equivalent to some element of $A$,
  similarly for $q_2^{-1}[A]$. From $\xi(x)\models\M q_1^{-1}[A]$ we
  conclude $\zeta(y)\models\M S[q_1^{-1}[A]]$ because $S$ is a
  $\Lambda$-simulation. But $S[q_1^{-1}[A]]\subseteq q_2^{-1}[A]$
  because clearly each element of $S[q_1^{-1}[A]]$ is
  $\sim_S$-equivalent to an element of $q_1^{-1}[A]$ and hence to an
  element of $A$. Therefore, $\zeta(y)\models\M q_2^{-1}[A]$. The
  converse implication is shown dually.

  For the second part, let $S$ be a $\Lambda$-bisimulation up to
  difunctionality. By Proposition~\ref{prop:difunctional}, the
  difunctional closure $\bar S$ of $S$ is a $\Lambda$-bisimulation and
  hence, by the first part, a $T$-bisimulation. By composing the
  $T$-coalgebra structure $\rho:\bar S\to T\bar S$ as in the definition of
  $T$-bisimulation with the inclusion $S\into\bar S$, we see that $S$
  is a $T$-bisimulation up to difunctionality.\qed
\end{proof}

\begin{cor}
  If $T$ preserves weak pullbacks, then $T$-bisimulations up to
  difunctionality are sound (and complete) for $T$-behavioural equivalence.
\end{cor}

\section{Conclusions}

We have introduced novel notions of $\Lambda$-simulation and
$\Lambda$-bisimulation that work well in a setting where the
coalgebraic type functor admits a separating set $\Lambda$ of monotone
predicate liftings. In particular, we have shown that
$\Lambda$-bisimilarity is, in this setting, always sound and complete
for $T$-behavioural equivalence, and moreover always admits a natural
notion of bisimulation up to difunctionality. We have shown that
$T$-bisimulations are always $\Lambda$-bisimulations, similarly for
versions up to difunctionality, and that the converse holds for
versions up to difunctionality in case $T$ preserves weak
pullbacks. We leave the question whether the converse holds in the
plain case under preservation of weak pullbacks as an open problem.

\bibliographystyle{splncs03}
\bibliography{coalgml}

\end{document}